\documentclass[conference,letterpaper]{IEEEtran}

\usepackage[utf8]{inputenc} 
\usepackage[T1]{fontenc}
\usepackage{url}
\usepackage{ifthen}
\usepackage{cite}
\usepackage[cmex10]{amsmath} 
\usepackage{bbm,bm}
\usepackage{amssymb}
\usepackage{pgfplots}
\usepackage{siunitx}
\usepackage{graphicx}
\usepackage{mathtools}

      \newtheorem{theorem}{Theorem}
      \newtheorem{remark}{Remark}
      \newtheorem{lemma}{Lemma}

      \newtheorem{proof}{Proof}

\begin{document}
\title{Lower bound on Wyner's Common Information } 



      \author{
      \IEEEauthorblockN{Erixhen Sula \\}
      \IEEEauthorblockA{\'Ecole Polytechnique F\'ed\'erale \\ de Lausanne\\
      Lausanne, Switzerland\\
      erixhen.sula@epfl.ch}
      \and
      \IEEEauthorblockN{Michael Gastpar \\}
      \IEEEauthorblockA{\'Ecole Polytechnique F\'ed\'erale \\ de Lausanne\\
      Lausanne, Switzerland\\
      michael.gastpar@epfl.ch}}%


\maketitle


\begin{abstract}
An important notion of common information between two random variables is due to Wyner. In this paper, we derive a lower bound on Wyner's common information for continuous random variables. The new bound improves on the only other general lower bound on Wyner's common information, which is the mutual information.
We also show that the new lower bound is tight for the so-called ``Gaussian channels'' case, namely, when the joint distribution of the random variables can be written as the sum of a single underlying random variable and Gaussian noises.
We motivate this work from the recent variations of Wyner's common information and applications to network data compression problems such as the Gray-Wyner network. 
\end{abstract}

\section{Introduction}
Extracting and assessing common features amongst multiple variables is a natural task occurring in many different problem settings. Wyner's common information~\cite{Wyner} provides one answer to this, which was originally defined for finite alphabets as follows
\begin{align} \label{eqn:Wynerdef}
C(X;Y)= \inf_{W: X - W - Y } I(X,Y;W).
\end{align}
For a pair of random variables, it seeks to find the most compact third variable that makes the pair conditionally independent. Compactness is measured in terms of the mutual information between the pair and the third variable.
In~\cite{Wyner}, Wyner also identifies two operational interpretations.
The first concerns a source coding network often referred to as the Gray-Wyner network.
For this scenario, Wyner's common information characterizes the smallest common rate required to enable two decoders to recover $X$ and $Y,$ respectively, in a lossless fashion.
The second operational interpretation pertains to the distributed simulation of common randomness. Here, Wyner's common information characterizes the smallest number of random bits that need to be shared between the processors.
In subsequent work, Wyner's common information was extended to continuous random variables and was computed for a pair of Gaussian random variables \cite{Xu--Liu--Chen,Xu--Liu--Chen-2} and for a pair of additive ``Gaussian channel'' distributions \cite{Yang-Chen14}. Other related works include \cite{Veld--Gastpar, Lapidoth--Wigger}. Wyner's common information has many applications, including to communication networks~\cite{Wyner}, to caching~\cite[Section III.C]{Wang--Lim--Gastpar}, to source coding \cite{Satpathy--Cuff}, and to feature extraction~\cite{SulaG:21entropy}.

In this paper, we derive a new lower bound on Wyner's common information for continuous random variables. The proof is based on a method known as factorization of convex envelopes, which was originally introduced in \cite{Geng--Nair}.
The proof strategy is fundamentally different from the techniques that were used to solve Wyner's original common information problem. Specifically, for the latter, the generic approach is to first characterize the class of variables that enable conditional independence, and then inside this class to find the optimal variable.
By contrast, we lower bound the Wyner's common information problem by a convex problem, which we can then solve via optimizing.

We illustrate the promise of the new lower bound by considering Gaussian mixture distributions and Laplace distributions.
We also establish that the new lower bound is tight for a simple case of the so-called ``Gaussian channels'' distribution. Here, $X$ and $Y$ can be written as the sum of a single arbitrary random variable and jointly Gaussian noises.
We note that for this special case, Wyner's common information was previously found, using different methods, in~\cite{Yang-Chen14}.

We use the following notation. Random variables are denoted by uppercase letters $X,Y,Z$ and their realizations by lowercase letters $x,y,z$.
For the cross-covariance matrix of $X$ and $Y$, we use the shorthand notation $K_{X Y}$, and for the covariance matrix of a random vector $ X$ we use the shorthand notation $K_{ X}:= K_{ X X}$. Let $p_X(x)$ denote the probability density function of random variable $X$ at realisation $x$. Let $\mathcal{N}(m,\sigma^2)$ be the Gaussian probability density function with mean $m$ and variance $\sigma^2$.

\section{Main Result}
Here we present our lower bound on Wyner's common information. The bound is given in terms of the entropy of the pair, entropy and Wyner's common information for Gaussian random variables. The theorem says:
\begin{theorem} \label{thm:lowerWyner}
Let $(X,Y)$ have probability density function $p_{(X,Y)}$ that satisfy the covariance constraint $K_{(X,Y)}$. Let, $(X_g,Y_g) \sim \mathcal{N}(0,K_{(X,Y)})$, then
\begin{align} \label{eqn:mainineq}
C(X;Y) \geq  \max \{ C(X_g;Y_g) +h(X,Y)-h(X_g,Y_g), 0 \}.
\end{align}
where 
\begin{align}
C(X_g;Y_g) =\frac{1}{2} \log{\frac{1+|\rho|}{1-|\rho|}},
\end{align}
and $\rho$ is the correlation coefficient between $X$ and $Y$.
\end{theorem}

The proof is given in Section \ref{sec:proofmainthm}. A similar argument is used for the max-entropy bound where the probability density functions have covariance constraints. Interestingly, once we plug in Gaussian random variables and additive ``Gaussian channel'' distributions, then the bound is attained with equality. 

\begin{remark}
In \cite{Wyner}, it is showed that $C(X;Y) \geq I(X;Y)$. In Section \ref{Sec-mixture-exact}-\ref{sec:Laplace} we show that our lower bound from Theorem \ref{thm:lowerWyner} can be tighter.
\end{remark}

\begin{remark}
The bound of Theorem~\ref{thm:lowerWyner} can be expressed equivalently as
\begin{align} \label{eqn:derivedineq}
C(X;Y)\geq  C(X_g;Y_g) -D\left(p_{(X,Y)}\left\|p_{(X_g,Y_g)}\right)\right. .
\end{align}
\end{remark}

\begin{remark}
The bound of Theorem~\ref{thm:lowerWyner} can be negative (if not for the correction). If we choose $X$ and $Y$ to be independent, then $X_g$ and $Y_g$ will be independent as well. Thus, the bound in (\ref{eqn:derivedineq}) becomes 
\begin{align}
C(X;Y)\geq  -D\left(p_{X}\left\|p_{X_g}\right) \right. - D\left(p_{Y}\left\|p_{Y_g}\right) \right. ,
\end{align}
that is a negative bound from the positivity of the Kullback-Leibler divergence.  
\end{remark}

In the latter section, we provide pairs of random variable and compute our lower bounds on Wyner's common information to verify the usefulness of the derived bound.

%
%
%

\section{Additive ``Gaussian Channel'' Distributions}\label{Sec-mixture-exact}

In this section, we consider the distributions that are described as follows.
Let $(\hat{X}, \hat{Y})$ be a Gaussian distribution with mean zero and covariance matrix
\begin{align}
K_{(\hat{X}, \hat{Y})}=
\begin{pmatrix}
1 & \hat{\rho}\\
\hat{\rho} & 1
\end{pmatrix}.
\end{align}
Then, we consider the two-dimensional source given by
\begin{align} \label{eqn:mixturerv}
\begin{pmatrix}
X\\
Y
\end{pmatrix} &=
\begin{pmatrix}
\hat{X}\\
\hat{Y}
\end{pmatrix} +
\begin{pmatrix} 
A\\
B
\end{pmatrix}.
\end{align}
Let $(A,B)$ be arbitrary random variables with mean zero and covariance 
\begin{align} \label{eqn:covAB}
K_{(A,B)}=
\begin{pmatrix}
\sigma_A^2 & r \sigma_A \sigma_B\\
r \sigma_A \sigma_B & \sigma_B^2
\end{pmatrix},
\end{align}
where $\sigma_A=\sigma_B$ and $(A,B)$ is independent of the pair $(\hat{X},\hat{Y})$. For this particular distribution, we evaluate our lower bound in (\ref{thm:lowerWyner}) and also provide an upper bound.

\subsection{Lower Bound}
We have that ${\mathbb E}[X]={\mathbb E}[Y]=0$ and 
\begin{align} \label{eqn:K_XY_comp}
{\mathbb E}[X^2] &= {\mathbb E}[\hat{X}^2]+ {\mathbb E}[A^2] =1+\sigma_A^2, \\
{\mathbb E}[XY] &= {\mathbb E}[\hat{X} \hat{Y}]+{\mathbb E}[AB] =\hat{\rho}+r \sigma_A^2.
\end{align}
By symmetry ${\mathbb E}[Y^2]={\mathbb E}[X^2]$ and 
\begin{align} \label{eqn:rho_XY_comp}
\rho=\frac{{\mathbb E}[XY] }{\sqrt{{\mathbb E}[X^2]{\mathbb E}[Y^2]}} =\frac{\hat{\rho}+r \sigma_A^2}{1+\sigma_A^2}.
\end{align}
Therefore, the formula given in Theorem~\ref{thm:lowerWyner} evaluates to
\begin{align}
C(X;Y)  &\geq C(X_g;Y_g) +h(X,Y)-h(X_g,Y_g) \\
&= \frac{1}{2}\log{\frac{1+\rho}{1-\rho}} + h(X,Y) \nonumber \\
& \hspace{1.2em} -\frac{1}{2}\log{(2 \pi e)^2 \left( (1+\sigma_A^2)^2 -(\hat{\rho}+r \sigma_A^2)^2 \right)} \label{eqn:expl_lowercomp}\\
&= h(X,Y) -\log{ \left( 2 \pi e (1-\hat{\rho}+(1-r)\sigma_A^2 \right)}. \label{eqn:LBAG}
\end{align}
where (\ref{eqn:expl_lowercomp}) follows from substituting for $K_{(X,Y)}$ and (\ref{eqn:LBAG}) follows from substituting for $\rho$ computed in (\ref{eqn:rho_XY_comp}).

\subsection{Upper Bound} \label{sec:UpperBound}
Next we give an upper bound on Wyner's common information for the example of this section.
To accomplish this, rewrite the pair $(\hat{X},\hat{Y})$ as
\begin{align} \label{eqn:gausssplit}
\hat{X} &= \sqrt{\hat{\rho}} V + Z_x, \nonumber \\
\hat{Y} &= \sqrt{\hat{\rho}} V + Z_y,
\end{align}
where $V,Z_x,Z_y$ are mutually independent, $V \sim \mathcal{N}(0,1)$ and $Z_x,Z_y \sim \mathcal{N}(0,1-\hat{\rho})$.
Then, a valid choice to make $X$ and $Y$ conditionally independent on $W$ is $W =(\sqrt{\hat{\rho}} V + A,\sqrt{\hat{\rho}} V + B)$.
By combining (\ref{eqn:mixturerv}) and (\ref{eqn:gausssplit}) we can rewrite the pair $(X,Y)$ as
\begin{align}
X=&\sqrt{\hat{\rho}} V + A+Z_x, \nonumber \\
Y=&\sqrt{\hat{\rho}} V + B+Z_y,
\end{align}
where $W$ is independent of $Z_x$ and $Z_y$.
So we have
\begin{align}
&I(X;Y|W) \nonumber \\
&=I(\sqrt{\hat{\rho}} V+A + Z_x;\sqrt{\hat{\rho}} V+B + Z_y| W) \\
&=I(Z_x;Z_y|W) \label{eqn:exp1lb}\\
&=I(Z_x;Z_y) \label{eqn:exp2lb}\\
&=0,\label{eqn:exp3lb}
\end{align}
where (\ref{eqn:exp1lb}) follows by subtracting the parts that are in the conditioning by recalling that $W=(\sqrt{\hat{\rho}} V+A,\sqrt{\hat{\rho}} V+B)$, (\ref{eqn:exp2lb}) follows from independence of $W$ and $(Z_x,Z_y)$ and (\ref{eqn:exp3lb}) follows from the independence of $Z_x$ and $Z_y$. 

Thus, the upper bound is
\begin{align}
&C(X;Y) \nonumber \\
&\le I(X,Y; W) \label{eqn:upexp1}\\
&= h(X,Y) - h(\sqrt{\hat{\rho}} V+A+Z_x,\sqrt{\hat{\rho}} V+B+Z_y|W) \label{eqn:upexp2}\\
&= h(X,Y) - h(Z_x,Z_y|W) \label{eqn:upexp3}\\
&= h(X,Y) - h(Z_x,Z_y) \label{eqn:upexp4}\\
&= h(X,Y) -\log{ \left( 2 \pi e (1-\hat{\rho}) \right) }. \label{eqn:UBAG}
\end{align}
where (\ref{eqn:upexp1}) follows from the definition of $C(X;Y)$ where $W$ satisfies $X-W-Y$, (\ref{eqn:upexp2}) follows by rewriting the mutual information, (\ref{eqn:upexp3}) follows from subtracting the parts that are in the conditioning and (\ref{eqn:upexp4}) follows from independence of $W$ and $(Z_x,Z_y)$.

\subsection{Example 1}
\begin{lemma} \label{lem:Gaussaddchannel}
For the additive ``Gaussian channel'' distributions described in (\ref{eqn:mixturerv}) and $A=B$, we have
\begin{align}
C(X;Y) &= h(X,Y) -\log{ \left( 2 \pi e (1-\hat{\rho}) \right)}.
\end{align}
\end{lemma}

The proof follows from the fact that the lower bound (\ref{eqn:LBAG}) and upper bound (\ref{eqn:UBAG}) coincide when $A=B$, which means $r=1$. The same result is derived by a different approach in \cite{Yang-Chen14}.
To illustrate Lemma \ref{lem:Gaussaddchannel}, let $A$ be binary $\pm \sigma_A$ with uniform probability. Then, we get Figure \ref{fig:Gaussaddchannel1}.

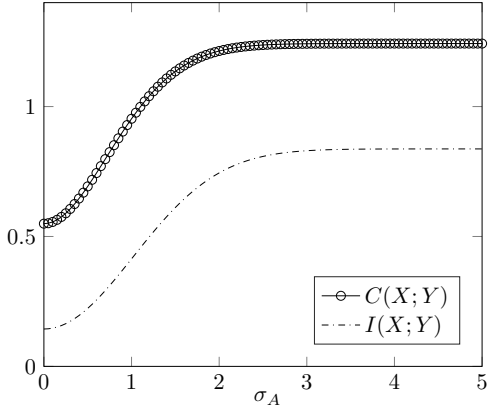
\begin{figure}[h!]
\centering
\scalebox{0.85}{
%
%
\begin{tikzpicture}

\draw[black,thick] (3.5,-0.5) node{$\sigma_A$};

\begin{axis}[%
xmin=0,
xmax=5,
ymin=0,
ymax=1.4,
axis background/.style={fill=white},
legend style={at={(0.95,0.2477)}, anchor=north east, legend cell align=left, align=left, draw=white!15!black}
]
\addplot [color=black, mark=o, mark options={solid, black}]
  table[row sep=crcr]{%
0	0.549306144334046\\
0.05	0.550970039370216\\
0.1	0.555928755209485\\
0.15	0.56408548528528\\
0.2	0.57528546087927\\
0.25	0.589324518997749\\
0.3	0.605959073762691\\
0.35	0.624916373711938\\
0.4	0.645904266119257\\
0.45	0.668620021725376\\
0.5	0.692758022493896\\
0.55	0.718016273194129\\
0.6	0.744101785085834\\
0.65	0.770734921165044\\
0.7	0.797652806726262\\
0.75	0.824611909307964\\
0.8	0.851389886130541\\
0.85	0.877786789092634\\
0.9	0.903625709328948\\
0.95	0.928752936081835\\
1	0.953037698413919\\
1.05	0.976371552985524\\
1.1	0.998667476536498\\
1.15	1.01985871760415\\
1.2	1.03989745818149\\
1.25	1.05875333230683\\
1.3	1.07641184486707\\
1.35	1.0928727301267\\
1.4	1.10814828563509\\
1.45	1.12226171321834\\
1.5	1.13524549476008\\
1.55	1.1471398264535\\
1.6	1.15799113122848\\
1.65	1.16785066516948\\
1.7	1.17677323000859\\
1.75	1.18481600025151\\
1.8	1.19203747022424\\
1.85	1.19849652335122\\
1.9	1.20425162332457\\
1.95	1.20936012451681\\
2	1.21387769703717\\
2.05	1.21785786023569\\
2.1	1.22135161720896\\
2.15	1.22440718194454\\
2.2	1.22706979013217\\
2.25	1.22938158434094\\
2.3	1.23138156418845\\
2.35	1.23310559226294\\
2.4	1.23458644688636\\
2.45	1.23585391327169\\
2.5	1.23693490521431\\
2.55	1.23785361012024\\
2.6	1.23863165089505\\
2.65	1.2392882589604\\
2.7	1.23984045341766\\
2.75	1.24030322211039\\
2.8	1.2406897010439\\
2.85	1.24101134927945\\
2.9	1.24127811703231\\
2.95	1.24149860525122\\
3	1.24168021545041\\
3.05	1.24182928898871\\
3.1	1.24195123536\\
3.15	1.24205064936129\\
3.2	1.24213141725946\\
3.25	1.24219681227191\\
3.3	1.24224957983378\\
3.35	1.24229201323057\\
3.4	1.24232602025605\\
3.45	1.2423531815948\\
3.5	1.24237480165434\\
3.55	1.24239195256561\\
3.6	1.24240551205883\\
3.65	1.24241619588733\\
3.7	1.24242458543896\\
3.75	1.24243115112461\\
3.8	1.24243627209184\\
3.85	1.24244025275515\\
3.9	1.24244333659223\\
3.95	1.24244571759974\\
4	1.24244754976393\\
4.05	1.24244895485022\\
4.1	1.24245002878393\\
4.15	1.24245084684999\\
4.2	1.24245146791433\\
4.25	1.24245193783277\\
4.3	1.24245229219469\\
4.35	1.24245255851875\\
4.4	1.24245275800532\\
4.45	1.24245290692654\\
4.5	1.24245301772681\\
4.55	1.24245309988771\\
4.6	1.24245316060779\\
4.65	1.2424532053317\\
4.7	1.24245323816332\\
4.75	1.24245326218389\\
4.8	1.24245327969943\\
4.85	1.24245329242859\\
4.9	1.24245330164856\\
4.95	1.2424533083042\\
5	1.24245331309285\\
};
\addlegendentry{$C(X;Y)$}

\addplot [color=black, dashdotted]
  table[row sep=crcr]{%
0	0.143841036225903\\
0.05	0.144674021381876\\
0.1	0.147168754601365\\
0.15	0.151312265057855\\
0.2	0.157082066070372\\
0.25	0.164445262603877\\
0.3	0.173357869949305\\
0.35	0.183764546040646\\
0.4	0.195598788017088\\
0.45	0.208783540972767\\
0.5	0.223232122435526\\
0.55	0.238849361863648\\
0.6	0.255532869686199\\
0.65	0.273174371100706\\
0.7	0.291661058823525\\
0.75	0.310876933691456\\
0.8	0.330704112229624\\
0.85	0.351024086818912\\
0.9	0.371718927919748\\
0.95	0.392672419874278\\
1	0.413771122839694\\
1.05	0.434905353937674\\
1.1	0.455970081084343\\
1.15	0.476865723422302\\
1.2	0.497498852922667\\
1.25	0.517782792615401\\
1.3	0.537638108040253\\
1.35	0.556992989856391\\
1.4	0.575783527056572\\
1.45	0.593953871840461\\
1.5	0.611456298842611\\
1.55	0.628251163026549\\
1.6	0.644306762077449\\
1.65	0.659599110510678\\
1.7	0.674111633908907\\
1.75	0.687834792680401\\
1.8	0.700765645468393\\
1.85	0.712907362825753\\
1.9	0.724268701997569\\
1.95	0.734863453633441\\
2	0.744709870996512\\
2.05	0.75383009176808\\
2.1	0.762249561893847\\
2.15	0.769996470107904\\
2.2	0.777101200839716\\
2.25	0.783595812191154\\
2.3	0.789513544593438\\
2.35	0.794888364661661\\
2.4	0.799754547671088\\
2.45	0.804146301028045\\
2.5	0.8080974301076\\
2.55	0.811641046910479\\
2.6	0.814809321158374\\
2.65	0.81763327271871\\
2.7	0.820142603625635\\
2.75	0.822365567453513\\
2.8	0.824328873395043\\
2.85	0.826057622100926\\
2.9	0.827575270138796\\
2.95	0.828903619826474\\
3	0.830062831167691\\
3.05	0.831071452670297\\
3.1	0.831946467931779\\
3.15	0.832703355039127\\
3.2	0.833356156022692\\
3.25	0.833917553833405\\
3.3	0.834398954552044\\
3.35	0.834810572798623\\
3.4	0.835161518561647\\
3.45	0.835459883925899\\
3.5	0.835712828416246\\
3.55	0.835926661911237\\
3.6	0.836106924289722\\
3.65	0.836258461174803\\
3.7	0.836385495309907\\
3.75	0.836491693262742\\
3.8	0.836580227281026\\
3.85	0.836653832245214\\
3.9	0.83671485775151\\
3.95	0.836765315443171\\
4	0.836806921761856\\
4.05	0.836841136344385\\
4.1	0.83686919631623\\
4.15	0.836892146762525\\
4.2	0.836910867662228\\
4.25	0.83692609758271\\
4.3	0.836938454422602\\
4.35	0.83694845349022\\
4.4	0.836956523186538\\
4.45	0.836963018553587\\
4.5	0.836968232926541\\
4.55	0.836972407916035\\
4.6	0.836975741922332\\
4.65	0.836978397371571\\
4.7	0.836980506838576\\
4.75	0.836982178211197\\
4.8	0.836983499026104\\
4.85	0.836984540099378\\
4.9	0.836985358551647\\
4.95	0.836986000323343\\
5	0.836986502255064\\
};
\addlegendentry{$I(X;Y)$}

\end{axis}
\end{tikzpicture}%
}
\vspace{-1em}
\caption{The o-line is the exact Wyner's common information $C(X;Y)$ for the specified Gaussian mixture distribution. The dashed line is the mutual information $I(X;Y)$. In this setup we plot $C(X;Y)$ and $I(X;Y)$ in nats versus $\sigma_A$ for $\hat{\rho}=0.5$.} \label{fig:Gaussaddchannel1}
\end{figure}
\subsection{Example 2}
Another example is to choose $(A,B)$ be doubly symmetric binary distribution where $p_{(A,B)}(A=B=\sigma_A)=p_{(A,B)}(A=B=-\sigma_A)=\frac{1+r}{4}$ and $p_{(A,B)}(A=-B=\sigma_A)=p_{(A,B)}(A=-B=-\sigma_A)=\frac{1-r}{4}$. 
Note that for these choices, the covariance matrix of $A$ and $B$ is given by Equation (\ref{eqn:covAB}). If we select $A=B$ or $r=1,$ this model is precisely the model studied in Example 1. A numerical evaluation is shown in Figure \ref{fig:Gaussaddchannel2}. 

\begin{figure}[h!]
\centering
\scalebox{0.85}{
%
%
\begin{tikzpicture}

\draw[black,thick] (3.5,-0.5) node{$\sigma_A$};
\begin{axis}[%
xmin=0,
xmax=3,
ymin=-0.5,
ymax=1.7,
axis background/.style={fill=white},
legend style={at={(0.85,0.2777)}, anchor=north east,legend cell align=left, align=left, draw=white!15!black}
]
\addplot [color=black, mark=asterisk, mark options={solid, black}]
  table[row sep=crcr]{%
0	0.549306144334046\\
0.05	0.550304148649307\\
0.1	0.553274417125354\\
0.15	0.558147190176099\\
0.2	0.564810805083345\\
0.25	0.573117652142203\\
0.3	0.582891002559507\\
0.35	0.593931848477728\\
0.4	0.606025171329101\\
0.45	0.618945360279453\\
0.5	0.632460750746126\\
0.55	0.646337407351311\\
0.6	0.660342337997021\\
0.65	0.67424631992845\\
0.7	0.68782647500311\\
0.75	0.70086867557933\\
0.8	0.713169811048747\\
0.85	0.724539906000075\\
0.9	0.734804056261102\\
0.95	0.743804137286509\\
1	0.751400237871626\\
1.05	0.757471778064229\\
1.1	0.761918280777891\\
1.15	0.764659779787079\\
1.2	0.765636860769694\\
1.25	0.764810345557754\\
1.3	0.762160641834492\\
1.35	0.75768679058199\\
1.4	0.751405251312427\\
1.45	0.743348470416801\\
1.5	0.733563280905141\\
1.55	0.722109182597862\\
1.6	0.709056550731589\\
1.65	0.694484818294259\\
1.7	0.678480673533942\\
1.75	0.661136309328973\\
1.8	0.642547755767619\\
1.85	0.622813321642976\\
1.9	0.60203216486319\\
1.95	0.580303006208401\\
2	0.557722995602786\\
2.05	0.534386735236746\\
2.1	0.510385459561907\\
2.15	0.485806368456904\\
2.2	0.460732106748072\\
2.25	0.435240380779208\\
2.3	0.409403700840309\\
2.35	0.383289236952725\\
2.4	0.356958774724165\\
2.45	0.330468757668621\\
2.5	0.303870402477811\\
2.55	0.277209874158018\\
2.6	0.250528508648028\\
2.65	0.223863071441745\\
2.7	0.197246041794584\\
2.75	0.170705913239323\\
2.8	0.144267502326064\\
2.85	0.11795225869191\\
2.9	0.0917785707217711\\
2.95	0.0657620621602388\\
3	0.0399158760482832\\
3.05	0.0142509432812097\\
3.1	-0.011223766097594\\
3.15	-0.0365010100422003\\
3.2	-0.0615750606569248\\
3.25	-0.0864415094230373\\
3.3	-0.111097091158437\\
3.35	-0.135539526051362\\
3.4	-0.159767378871445\\
3.45	-0.183779934291035\\
3.5	-0.207577087141946\\
3.55	-0.231159246373778\\
3.6	-0.25452725146336\\
3.65	-0.277682300039897\\
3.7	-0.300625885533305\\
3.75	-0.323359743711557\\
3.8	-0.345885807048049\\
3.85	-0.368206165939671\\
3.9	-0.390323035883101\\
3.95	-0.412238729802601\\
4	-0.43395563480861\\
4.05	-0.455476192746554\\
4.1	-0.47680288397435\\
4.15	-0.497938213876063\\
4.2	-0.518884701687131\\
4.25	-0.539644871262464\\
4.3	-0.560221243474855\\
4.35	-0.580616329974374\\
4.4	-0.600832628083793\\
4.45	-0.620872616637202\\
4.5	-0.640738752602963\\
4.55	-0.660433468355752\\
4.6	-0.679959169487219\\
4.65	-0.699318233061545\\
4.7	-0.7185130062403\\
4.75	-0.737545805212352\\
4.8	-0.756418914377441\\
4.85	-0.775134585739734\\
4.9	-0.793695038476802\\
4.95	-0.812102458654116\\
5	-0.830358999061937\\
};
\addlegendentry{Lower bound from Theorem \ref{thm:lowerWyner}}

\addplot [color=black, dashdotted]
  table[row sep=crcr]{%
0	0.143841036225903\\
0.05	0.144340411769701\\
0.1	0.145831071415959\\
0.15	0.148290818795564\\
0.2	0.151683372710006\\
0.25	0.155959516869052\\
0.3	0.161058797315198\\
0.35	0.166911741860697\\
0.4	0.173442491887791\\
0.45	0.180571663761619\\
0.5	0.188219214378971\\
0.55	0.196307083616444\\
0.6	0.204761423817406\\
0.65	0.213514289847368\\
0.7	0.222504735018237\\
0.75	0.2316793234234\\
0.8	0.240992119265411\\
0.85	0.25040424595442\\
0.9	0.259883123472778\\
0.95	0.269401495000275\\
1	0.278936346760774\\
1.05	0.288467811841584\\
1.1	0.297978132097761\\
1.15	0.307450734292448\\
1.2	0.316869458898775\\
1.25	0.326217963582776\\
1.3	0.335479309012942\\
1.35	0.344635722699949\\
1.4	0.353668527221413\\
1.45	0.362558212391055\\
1.5	0.371284626525159\\
1.55	0.379827259667979\\
1.6	0.388165591149122\\
1.65	0.396279474811475\\
1.7	0.4041495373368\\
1.75	0.411757567984734\\
1.8	0.419086881467694\\
1.85	0.426122639357344\\
1.9	0.432852119149602\\
1.95	0.43926492373607\\
2	0.44535312740346\\
2.05	0.4511113575132\\
2.1	0.456536813635471\\
2.15	0.461629228079773\\
2.2	0.466390773471544\\
2.25	0.470825924270084\\
2.3	0.474941279933078\\
2.35	0.478745357844159\\
2.4	0.482248364174553\\
2.45	0.485461950606054\\
2.5	0.488398964348727\\
2.55	0.491073198205968\\
2.6	0.493499146621464\\
2.65	0.4956917727409\\
2.7	0.497666290580447\\
2.75	0.499437965454052\\
2.8	0.501021934906753\\
2.85	0.502433051555795\\
2.9	0.503685748478506\\
2.95	0.504793927113447\\
3	0.505770867074768\\
3.05	0.506629156813649\\
3.1	0.507380643698526\\
3.15	0.50803640182243\\
3.2	0.508606715669214\\
3.25	0.509101077677474\\
3.3	0.509528197715331\\
3.35	0.509896022512806\\
3.4	0.510211763180783\\
3.45	0.510481929061491\\
3.5	0.510712366300862\\
3.55	0.510908299694024\\
3.6	0.511074376526794\\
3.65	0.511214711308862\\
3.7	0.511332930466785\\
3.75	0.51143221622662\\
3.8	0.51151534907116\\
3.85	0.511584748294795\\
3.9	0.511642510306249\\
3.95	0.511690444438829\\
4	0.511730106124272\\
4.05	0.511762827365026\\
4.1	0.511789744508239\\
4.15	0.511811823375624\\
4.2	0.511829881847428\\
4.25	0.511844610026512\\
4.3	0.511856588133766\\
4.35	0.511866302296364\\
4.4	0.511874158401415\\
4.45	0.511880494185686\\
4.5	0.511885589733513\\
4.55	0.511889676546238\\
4.6	0.511892945341258\\
4.65	0.511895552726569\\
4.7	0.511897626888249\\
4.75	0.51189927241481\\
4.8	0.511900574372735\\
4.85	0.511901601734368\\
4.9	0.511902410250229\\
4.95	0.51190304484532\\
5	0.511903541611482\\
};
\addlegendentry{$I(X;Y)$}

\addplot [color=black, mark=diamond, mark options={solid, black}]
  table[row sep=crcr]{%
0	0.549306144334046\\
0.05	0.55130364898239\\
0.1	0.557266438394892\\
0.15	0.567106931547571\\
0.2	0.580684154239635\\
0.25	0.597810264732574\\
0.3	0.618258146396798\\
0.35	0.641769177891888\\
0.4	0.668060562248554\\
0.45	0.696831898936525\\
0.5	0.727770930550451\\
0.55	0.760558551441333\\
0.6	0.794873230954627\\
0.65	0.830395002418381\\
0.7	0.86680913053155\\
0.75	0.90380951957602\\
0.8	0.941101879094754\\
0.85	0.978406629957125\\
0.9	1.01546151377592\\
0.95	1.05202386095584\\
1	1.08787247449284\\
1.05	1.12280909508161\\
1.1	1.15665942552308\\
1.15	1.189273706734\\
1.2	1.22052685220538\\
1.25	1.25031816133945\\
1.3	1.27857064389438\\
1.35	1.30522999728314\\
1.4	1.33026328547024\\
1.45	1.35365737266774\\
1.5	1.37541716707754\\
1.55	1.39556372981207\\
1.6	1.41413230215681\\
1.65	1.43117030086869\\
1.7	1.44673532658069\\
1.75	1.46089322494718\\
1.8	1.47371623422494\\
1.85	1.48528124681962\\
1.9	1.49566820617254\\
1.95	1.50495865441418\\
2	1.51323444063022\\
2.05	1.52057659449057\\
2.1	1.52706436546733\\
2.15	1.53277442397267\\
2.2	1.53778021750034\\
2.25	1.54215147226201\\
2.3	1.54595382884881\\
2.35	1.54924859908044\\
2.4	1.5520926303829\\
2.45	1.55453826369368\\
2.5	1.55663337097318\\
2.55	1.55842145882476\\
2.6	1.55994182543196\\
2.65	1.56122975893821\\
2.7	1.56231676646285\\
2.75	1.56323082410985\\
2.8	1.56399663953219\\
2.85	1.56463591982458\\
2.9	1.5651676386926\\
2.95	1.56560829796424\\
3	1.56597217954333\\
3.05	1.56627158484536\\
3.1	1.56651705959325\\
3.15	1.56671760257798\\
3.2	1.56688085761294\\
3.25	1.56701328842784\\
3.3	1.56712033667049\\
3.35	1.56720656351639\\
3.4	1.56727577563692\\
3.45	1.56733113645921\\
3.5	1.56737526376973\\
3.55	1.56741031478283\\
3.6	1.56743805982176\\
3.65	1.56745994575327\\
3.7	1.56747715028208\\
3.75	1.56749062816073\\
3.8	1.5675011503017\\
3.85	1.56750933670557\\
3.9	1.56751568403749\\
3.95	1.56752058860408\\
4	1.56752436540151\\
4.05	1.56752726382958\\
4.1	1.56752948059192\\
4.15	1.56753117023689\\
4.2	1.56753245372913\\
4.25	1.56753342538897\\
4.3	1.56753415848353\\
4.35	1.56753470971261\\
4.4	1.56753512279045\\
4.45	1.56753543129444\\
4.5	1.56753566091984\\
4.55	1.5675358312575\\
4.6	1.56753595718887\\
4.65	1.5675360499767\\
4.7	1.56753611811365\\
4.75	1.56753616798028\\
4.8	1.5675362043528\\
4.85	1.5675362307936\\
4.9	1.56753624994998\\
4.95	1.56753626378222\\
5	1.56753627373643\\
};
\addlegendentry{Upper Bound from Section \ref{sec:UpperBound}}

\end{axis}
\end{tikzpicture}%
}
\vspace{-1em}
\caption{The $*$-line is the lower bound on $C(X;Y)$ from Theorem \ref{thm:lowerWyner} and the $\diamond$-line is the upper bound on $C(X;Y)$ from Section \ref{sec:UpperBound}. The dashed line is the mutual information $I(X;Y)$. In this setup we plot the bounds on $C(X;Y)$ in nats versus $\sigma_A$ for $\hat{\rho}=0.5$ and $r=0.9$.} \label{fig:Gaussaddchannel2}
\end{figure}
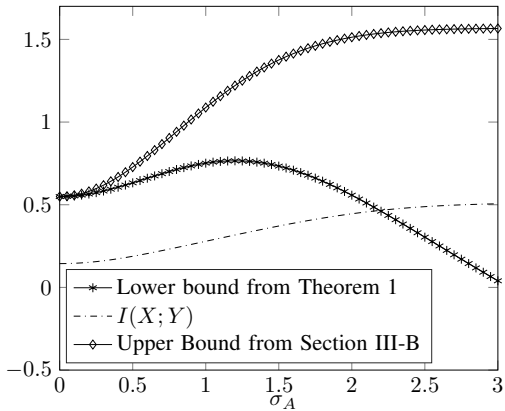

\section{Laplace Distributions} \label{sec:Laplace}
In this section, we consider the case when $(X,Y)$ is distributed according to the bivariate Laplace distribution described \cite[Section~5.1.3]{Kotz01} by
\begin{align}
p_{(X,Y)}(x,y)=\frac{1}{\pi \sqrt{1-\rho_{\ell}^2}} K_0 \left( \sqrt{\frac{2(x^2-2 \rho_{\ell}xy+y^2)}{1-\rho_{\ell}^2}} \right),
\end{align}
where $K_0$ is the modified Bessel function of the second kind described by
\begin{align}
K_0(z)=\frac{1}{2}\int_{-\infty}^{\infty} \frac{e^{i z t}}{\sqrt{t^2+1}} dt.
\end{align}
The variances of $X$ and $Y$ are unity and the correlation coefficient is $\rho_{\ell}$.
Define the entropy power of $(X,Y)$ as
\begin{eqnarray}
N(X,Y) &= \frac{1}{2\pi e} \exp(h(X,Y)).
\end{eqnarray}
Then, the bound of Theorem~\ref{thm:lowerWyner} can be expressed as
\begin{align}
 C(X;Y) &\ge \log \frac{N(X,Y)}{1-\rho_{\ell}}.
 \end{align}
Computation of the joint entropy $h(X,Y)$ as well as the mutual information $I(X;Y)$
leads to the curves in Figure~\ref{fig:Laplace}, further illustrating the potential of the new bound.


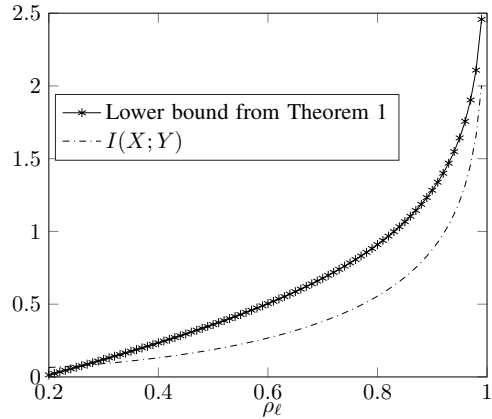
\begin{figure}[h!]
\centering
\scalebox{0.85}{
%
%
\begin{tikzpicture}

\draw[black,thick] (3.5,-0.5) node{$\rho_{\ell}$};
\begin{axis}[%
xmin=0.2,
xmax=1,
ymin=0,
ymax=2.5,
axis background/.style={fill=white},
legend style={at={(0.8,0.7777)}, legend cell align=left, align=left, draw=white!15!black}
]
\addplot [color=black, mark=asterisk, mark options={solid, black}]
  table[row sep=crcr]{%
0.2	0.0135917070445484\\
0.21	0.024028085213679\\
0.22	0.0345152264546752\\
0.23	0.0450485830787132\\
0.24	0.0556332278082955\\
0.25	0.0662713770130116\\
0.26	0.0769675601260176\\
0.27	0.0877229761114586\\
0.28	0.0985410808414728\\
0.29	0.10942529978651\\
0.3	0.120378646332692\\
0.31	0.131404457363081\\
0.32	0.142506162629595\\
0.33	0.153687288598641\\
0.34	0.164951567997924\\
0.35	0.176302798059586\\
0.36	0.187744949629327\\
0.37	0.199282152712875\\
0.38	0.21091873790879\\
0.39	0.222659127805758\\
0.4	0.234508028982328\\
0.41	0.246470327398553\\
0.42	0.258551082184467\\
0.43	0.270755711958587\\
0.44	0.283089687043555\\
0.45	0.29555925195165\\
0.46	0.308170391389401\\
0.47	0.320929463994386\\
0.48	0.333843428630412\\
0.49	0.34691948925427\\
0.5	0.36016530246227\\
0.51	0.373588932869159\\
0.52	0.387198923759387\\
0.53	0.40100433376755\\
0.54	0.415014781691614\\
0.55	0.429240497271653\\
0.56	0.443692374951742\\
0.57	0.45838201599863\\
0.58	0.473321874207494\\
0.59	0.488525252368481\\
0.6	0.50400632324824\\
0.61	0.519780500574599\\
0.62	0.535864227778184\\
0.63	0.552275306974154\\
0.64	0.569032933089574\\
0.65	0.586157895803815\\
0.66	0.603672822596016\\
0.67	0.621602317252613\\
0.68	0.639973231045123\\
0.69	0.658814949025943\\
0.7	0.678159722489549\\
0.71	0.698043024694162\\
0.72	0.718504179459194\\
0.73	0.739586522470962\\
0.74	0.761338536961799\\
0.75	0.783814228701487\\
0.76	0.807074233979449\\
0.77	0.831186907090974\\
0.78	0.856229693950995\\
0.79	0.882290845976568\\
0.8	0.90947147053402\\
0.81	0.937888207392332\\
0.82	0.96767664517807\\
0.83	0.998995583926663\\
0.84	1.03203269578781\\
0.85	1.06701198819192\\
0.86	1.10420384618413\\
0.87	1.14393882743593\\
0.88	1.18662682970306\\
0.89	1.23278504095418\\
0.9	1.28307866076342\\
0.91	1.33838359541597\\
0.92	1.39988608224379\\
0.93	1.46924920506839\\
0.94	1.54890853673109\\
0.95	1.64263999550043\\
0.96	1.75676933617745\\
0.97	1.90315490382808\\
0.98	2.10841910197775\\
0.99	2.45751158566533\\
};
\addlegendentry{Lower bound from Theorem \ref{thm:lowerWyner}}

\addplot [color=black, dashdotted]
  table[row sep=crcr]{%
0.2	0.0648219583894236\\
0.21	0.0669643624271532\\
0.22	0.0692162469635869\\
0.23	0.0715862951754564\\
0.24	0.0740737320132272\\
0.25	0.0766808095585318\\
0.26	0.0794076467776663\\
0.27	0.0822578828480038\\
0.28	0.0852331002503255\\
0.29	0.0883351232800282\\
0.3	0.0915664117258026\\
0.31	0.0949293381475136\\
0.32	0.0984264323021513\\
0.33	0.102060392118247\\
0.34	0.105833990083504\\
0.35	0.10975023215263\\
0.36	0.113812267118854\\
0.37	0.118023421003446\\
0.38	0.122387177153972\\
0.39	0.126907308128784\\
0.4	0.131587708903425\\
0.41	0.136432528803581\\
0.42	0.141446207376967\\
0.43	0.146633320314716\\
0.44	0.151998922329149\\
0.45	0.157547862923733\\
0.46	0.163285862154178\\
0.47	0.169218922561345\\
0.48	0.175353152896014\\
0.49	0.181695178129258\\
0.5	0.188251992217438\\
0.51	0.195031069128068\\
0.52	0.202040365440576\\
0.53	0.209288364630245\\
0.54	0.216784121927144\\
0.55	0.224537313065881\\
0.56	0.23255829123785\\
0.57	0.240858168415661\\
0.58	0.249448807616992\\
0.59	0.258342981035065\\
0.6	0.267554522745678\\
0.61	0.277098153403608\\
0.62	0.286989912603284\\
0.63	0.297247080489475\\
0.64	0.307888428562169\\
0.65	0.318934342814625\\
0.66	0.330406952895676\\
0.67	0.34233042138876\\
0.68	0.354731166263004\\
0.69	0.367638146596764\\
0.7	0.381083195956149\\
0.71	0.395101445427218\\
0.72	0.409731610473456\\
0.73	0.425016911632562\\
0.74	0.441005225124573\\
0.75	0.457750246538166\\
0.76	0.475312235780459\\
0.77	0.49375917708773\\
0.78	0.513168152798543\\
0.79	0.533627016407863\\
0.8	0.555236556019843\\
0.81	0.578113113549081\\
0.82	0.602391897033619\\
0.83	0.628231372124975\\
0.84	0.655818882080267\\
0.85	0.685378110813721\\
0.86	0.717179124308464\\
0.87	0.751552115210384\\
0.88	0.788906820616792\\
0.89	0.829759986355305\\
0.9	0.874776546350393\\
0.91	0.924832127355661\\
0.92	0.981112676184223\\
0.93	1.04528094598415\\
0.94	1.11977229414871\\
0.95	1.20836239217333\\
0.96	1.31737660281051\\
0.97	1.45867310761167\\
0.98	1.65887401757016\\
0.99	2.00292871444252\\
};
\addlegendentry{$I(X;Y)$}

\end{axis}
\end{tikzpicture}%
}
\vspace{-1em}
\caption{The $*$-line is the lower bound on $C(X;Y)$ from Theorem \ref{thm:lowerWyner} and the dashed line is the mutual information $I(X;Y)$ for the described Laplace distribution. In this setup we plot the bounds on $C(X;Y)$ in nats versus $\rho_{\ell}$.} \label{fig:Laplace}
\end{figure}

\section{Proof of Theorem \ref{thm:lowerWyner}} \label{sec:proofmainthm}
\subsection{Preliminary}
\begin{theorem}[Theorem~2 in \cite{Hyper_Gauss}] \label{Thm:Hypercontract}
For $K \succeq 0$, $0 < \lambda < 1$, there exists a $0\preceq K^{\prime} \preceq K$ and $(X^{\prime},Y^{\prime})\sim \mathcal{N}(0,K^{\prime})$ such that $(X, Y)$ have distribution $p_{(X,Y)}$ with covariance constraint $K$, the following inequality holds
\begin{align}
\inf_W h(Y|W)+ &h(X|W) - (1+\lambda) h(X, Y|W)  \nonumber \\
&\geq h(Y^{\prime})+ h(X^{\prime}) -(1+\lambda)h(X^{\prime}, Y^{\prime}) . \label{Eq-Thm:Hypercontract}
\end{align}
\end{theorem}
\begin{proof}
The theorem is a consequence of \cite[Theorem~2]{Hyper_Gauss}, for a specific choice of $p=\frac{1}{\lambda}+1$. The proof regarding the existence of the infimum that is missing in \cite{Hyper_Gauss} is given in \cite{SulaG:19it}. 
\end{proof}

Before we jump into details it is important to realise that $\inf_W h(Y|W)+ h(X|W) - (1+\lambda) h(X, Y|W)$ is indeed the lower convex envelope of $h(Y)+ h(X) - (1+\lambda) h(X, Y)$ by thinking of $W$ as a time sharing random variable. In other words, we are taking the infimum over all convex envelopes such that for a covariance constraint on the pair $(X,Y)$ it satisfies the following 
\begin{align}
&\inf_{(X,Y)} \inf_W h(Y|W)+ h(X|W) - (1+\lambda) h(X, Y|W) \nonumber \\
&\quad \quad = \inf_{(X,Y)} h(Y)+ h(X) - (1+\lambda) h(X, Y).
\end{align}

The next lemma that is an optimization problem on the covariance matrix constraint for Gaussian random variables is needed for the proof of the theorems.
%

\begin{lemma} \label{lem:lemmacontractivity}
For $(X^{\prime},Y^{\prime})\sim \mathcal{N}(0,K^{\prime})$, the following inequality holds
\begin{align}
&\min_{K^{\prime}: 0 \preceq K^{\prime} \preceq \begin{pmatrix} 1 & \rho \\ \rho &1  \end{pmatrix}}  h(X^{\prime})+h(Y^{\prime})-(1+\lambda)h(X^{\prime},Y^{\prime}) \nonumber \\
& \quad \quad \geq  \frac{1}{2} \log{\frac{1}{1-\lambda^2}}-\frac{\lambda}{2} \log{(2\pi e)^2\frac{(1-\rho)^2(1+\lambda)}{1-\lambda}},
\end{align}
where $\lambda \leq \rho$.
\end{lemma}

\begin{proof}
The proof outline is given in Appendix \ref{App:lowerboundWCI}. For the full proof, refer to \cite{SulaG:19it}.
\end{proof}

\subsection{Lower bound on (relaxed) Wyner's common information}

Here, we consider a slightly more general case that is, we give a lower bound on relaxed Wyner's common information in Theorem \ref{thm:lowerWynerelaxed}. Thus, as a special case we obtain Theorem \ref{thm:lowerWyner}.
Let us define the relaxed Wyner's common information as in \cite{GastparS:19itw,SulaG:19it}.
For jointly continuous random variables $X$ and $Y$ with joint distribution $p(x,y),$ we define
\begin{align}
C_{\gamma} (X; Y) &=  \inf_{W:I(X;Y|W) \le \gamma} I(X,Y ; W), \label{Eq-def-Wyner-relaxed}
\end{align}
where the constraint of conditional independence is relaxed into an upper bound on the conditional mutual information. For $\gamma=0,$ we have $C_0(X;Y) = C(X;Y),$ the standard Wyner's common information.
A lower bound on relaxed Wyner's common information is given in the following theorem.

\begin{theorem} \label{thm:lowerWynerelaxed}
Let $(X,Y)$ have probability density functions $p_{(X,Y)}$ that satisfy the covariance constraint $K_{(X,Y)}$. Let, $(X_g,Y_g) \sim \mathcal{N}(0,K_{(X,Y)})$, then
\begin{align}
C_{\gamma}(X;Y)\geq  \max \{ C_{\gamma}(X_g;Y_g) +h(X,Y)-h(X_g,Y_g),0 \},
\end{align}
where
\begin{align}
C_{\gamma}(X_g;Y_g) &= \frac{1}{2} \log^+ \left( \frac{1 + |\rho|}{1-|\rho|} \cdot \frac{1 - \sqrt{1-e^{-2\gamma}}}{1 + \sqrt{1-e^{-2\gamma}}} \right),
\end{align}
and $\rho$ is the correlation coefficient between $X$ and $Y$.

\end{theorem}

\begin{proof}
Note that the mean of the random variables does not affect the Wyner's common information and its relaxed variant thus, we assume mean zero for both $X$ and $Y$. Also, the relaxed Wyner's common information is invariant to scaling of $X$ and $Y$.
Thus, without loss of generality we assume $X$ and $Y$ to be mean zero, unity variance and correlation coefficient $\rho$, so we proceed as follows
\begin{align}
&C_{\gamma}(X;Y) \nonumber \\
&=\inf_{W:I(X;Y|W) \leq \gamma} I(X,Y;W) \label{eqn:defWynerCI} \\
& \geq \inf_{W} (1+\mu)I(X,Y;W)-\mu I(X;W) -\mu I(Y;W)   \nonumber \\
& \quad \quad  +\mu I(X;Y) - \mu \gamma \label{eqn:alllambda} \\
&=  \mu \inf_{W}h(X|W)+h(Y|W) -(1+\frac{1}{\mu})h(X,Y|W)  \nonumber \\
& \quad \quad +h(X,Y) -\mu \gamma \label{eqn:rewritealllambda} \\
&\geq  \mu \hspace{-2em} \min_{K^{\prime}: 0 \preceq K^{\prime} \preceq \begin{pmatrix} 1 & \rho \\ \rho &1  \end{pmatrix}} h(X^{\prime})+h(Y^{\prime})-(1+\frac{1}{\mu})h(X^{\prime},Y^{\prime}) \nonumber \\
& \quad \quad +h(X,Y) -\mu \gamma \label{eqn:thm2sim} \\
& \geq h(X,Y) + \frac{\mu}{2} \log{\frac{\mu^2}{\mu^2-1}} \nonumber \\
& \quad \quad -\frac{1}{2} \log{(2\pi e)^2\frac{(1-\rho)^2(\mu+1)}{\mu-1}} -\mu \gamma  \label{eqn:lastexp} \\
& \geq h(X,Y) -h(X_g,Y_g) +C(X_g;Y_g) \label{eqn:lastexp2}
\end{align}
where (\ref{eqn:alllambda}) follows from weak duality and the bound is valid for all $\mu \geq 0$; (\ref{eqn:rewritealllambda}) follows from simplification; (\ref{eqn:thm2sim}) follows from Theorem \ref{Thm:Hypercontract} under the assumption that $\mu > 1$ where $(X^{\prime},Y^{\prime}) \sim \mathcal{N}(0,K^{\prime})$; 
(\ref{eqn:lastexp}) follows from Lemma \ref{lem:lemmacontractivity} under the assumption $\mu \geq \frac{1}{\rho}$ and (\ref{eqn:lastexp2}) follows by maximizing the function 
\begin{align}
g(\mu)&=h(X,Y) -\mu \gamma + \frac{\mu}{2} \log{\frac{\mu^2}{\mu^2-1}} \nonumber \\
& \quad \quad -\frac{1}{2} \log{(2\pi e)^2\frac{(1-\rho)^2(\mu+1)}{\mu-1}},
\end{align} 
for $\mu \geq \frac{1}{\rho}$. Now we need to solve $\max_{\mu \geq \frac{1}{\rho}}  g(\mu).$
The function $g$ is concave in $\mu$,
\begin{align}
\frac{\partial^2 g}{\partial \mu^2}&=-\frac{1}{\mu(\mu^2-1)} < 0,
\end{align}
and by studying the monotonicity we obtain
\begin{align}
\frac{\partial g}{\partial \mu}&=-\frac{1}{2}\log{\frac{\mu^2-1}{\mu^2}}  -\gamma.
\end{align}
Since the function is concave, the maximum is attained when the first derivative vanishes. That leads to the optimal solution $\mu_*=\frac{1}{\sqrt{1-e^{-2\gamma}}}$, where $\mu_*$ has to satisfy $\mu_* \geq \frac{1}{\rho}$. Substituting for the optimal solution we get 
\begin{align}
C_{\gamma}(X;Y) &\geq g \left( \frac{1}{\sqrt{1-e^{-2\gamma}}} \right) \\
&=h(X,Y) -h(X_g,Y_g)+C_{\gamma}(X_g;Y_g).
\end{align}
\end{proof}

\section{Vector Wyner's common information}

It is well-known that for $n$ independent pairs of random variables, we have
\begin{align}
C (X^n; Y^n) =  \sum_{i=1}^n C(X_i; Y_i).  \label{Eqn-thm:gensplit}
\end{align}
For the proof see \cite[Lemma~2]{SulaG:19it} by letting $\gamma=0$.

By making use of Theorem \ref{thm:lowerWyner} and (\ref{Eqn-thm:gensplit}) we can lower bound the Wyner's common information for $n$ independent pairs of random variables as
\begin{align}
C (X^n; Y^n) \geq  \sum_{i=1}^n  C(X_{g_i};Y_{g_i}) +h(X_i,Y_i)-h(X_{g_i},Y_{g_i}).
\end{align}

An interesting problem is finding a bound for arbitrary $(X^n,Y^n)$, for any dependencies between $X^n$ and $Y^n$. 
This is not studied here and is left for future investigation.

\appendices

\section{Proof Outline of Lemma \ref{lem:lemmacontractivity}} \label{App:lowerboundWCI}

Let us parametrize $K^{\prime}$ as $K^{\prime} = \begin{pmatrix} \sigma^2_X & q\sigma_X \sigma_Y \\ q\sigma_X \sigma_Y & \sigma^2_Y  \end{pmatrix} \succeq 0$. By substituting we obtain

\begin{align}
&\min_{K^{\prime}: 0 \preceq K^{\prime} \preceq \begin{pmatrix} 1 & \rho \\ \rho &1  \end{pmatrix}} h(X^{\prime})+h(Y^{\prime})-(1+\lambda)h(X^{\prime},Y^{\prime}) \nonumber \\
& \quad \quad \quad = \min_{(\sigma_X,\sigma_Y,q) \in \mathcal{A}_{\rho}} \frac{1}{2}\log{(2 \pi e)^2 \sigma_X^2\sigma_Y^2} \\
& \quad \quad \quad   -\frac{1+\lambda}{2}\log{(2 \pi e)^2 \sigma_X^2\sigma_Y^2(1-q^2)} \label{eqn:mlproof1}
\end{align}
where the set $\mathcal{A}_{\rho}$ is 
\begin{align} 
\mathcal{A}_{\rho}=\left\{( \sigma_X,\sigma_Y,q): \begin{pmatrix} \sigma^2_X -1& q\sigma_X \sigma_Y -\rho\\ q\sigma_X \sigma_Y-\rho & \sigma^2_Y -1 \end{pmatrix} \preceq 0 \right\}.
\end{align}

Another way of rewriting $\mathcal{A}_{\rho}$ is 

\begin{align} 
\mathcal{A}_{\rho}=\left\{(\sigma_X,\sigma_Y,q): \hspace{-1.5em} \substack{\sigma^2_X+\sigma^2_Y \leq 2, \\ \quad (1-q^2)\sigma^2_X\sigma^2_Y +2\rho q \sigma_X\sigma_Y +1-\rho^2-(\sigma^2_X+\sigma^2_Y) \geq 0} \right\}.
\end{align}

Let us define
\begin{align} 
\mathcal{B}_{\rho}=\left\{ (\sigma_X,\sigma_Y,q): \hspace{-1.5em} \substack{\sigma_X\sigma_Y \leq 1, \\ \quad (1-q^2)\sigma^2_X\sigma^2_Y +2\rho q \sigma_X\sigma_Y +1-\rho^2-2\sigma_X\sigma_Y) \geq 0} \right\},
\end{align} 
and the inequality $\sigma^2_X+\sigma^2_Y \geq 2\sigma_X\sigma_Y$, implies that $\mathcal{A}_{\rho} \subseteq \mathcal{B}_{\rho}$.
By reparametrizing $\sigma^2=\sigma_X\sigma_Y$, the set $\mathcal{B}_{\rho}$ becomes  
\begin{align} 
\mathcal{D}_{\rho}=\left\{( \sigma^2,q): \substack{\sigma^2 \leq 1, \\ (\sigma^2(1-q)-1+\rho)(\sigma^2(1+q)-1-\rho) \geq 0} \right\}.
\end{align}

The set $\mathcal{D}_{\rho}$ is rewritten as
\begin{align} 
\mathcal{D}_{\rho}=\left\{( \sigma^2,q): \substack{\text{for } \rho \geq q, \quad \sigma^2(1-q) \leq 1-\rho \\ \text{for } \rho < q, \quad  \sigma^2(1+q) \leq 1+\rho } \label{eqn:D} \right\}.
\end{align}

Thus, we have
\begin{align}
&\min_{(\sigma_X,\sigma_Y,q) \in \mathcal{A}_{\rho}} \frac{1}{2}\log{(2 \pi e)^2 \sigma_X^2\sigma_Y^2} \\
& -\frac{1+\lambda}{2}\log{(2 \pi e)^2 \sigma_X^2\sigma_Y^2(1-q^2)} \geq \min_{(\sigma^2,q) \in \mathcal{D}_{\rho}} f(\lambda,\sigma^2,q)
\label{eqn:mlproof2}
\end{align}
where,
\begin{align}f(\lambda,\sigma^2,q)&=\frac{1}{2}\log{(2 \pi e)^2 \sigma^4}-\frac{1+\lambda}{2}\log{(2 \pi e)^2 \sigma^4(1-q^2)} \label{eqn:f}.
\end{align}

\begin{itemize}

\item Let us consider the case $\rho \geq q$ for $\rho$ is positive. Then, by weak duality we have 
\begin{align}
&\min\limits_{(\sigma^2,q) \in \mathcal{D}_{\rho}} f(\lambda,\sigma^2,q) \nonumber \\
& \quad \quad \geq \min\limits_{\sigma^2,q} f(\lambda,\sigma^2,q) + \mu(\sigma^2(1-q)-1+\rho), \label{eqn:mlproof3}
\end{align}
for any $\mu \geq 0$.
By applying Karush-Kuhn-Tucker (KKT) conditions on (\ref{eqn:mlproof3}) we get
\begin{align}
\frac{\partial }{\partial \sigma^2}=-\frac{\lambda}{\sigma^2} + \mu (1-q)&=0, \label{eqn:KKT1} \\ 
\frac{\partial }{\partial q}=\frac{(1+\lambda)q}{1-q^2} - \mu \sigma^2&=0, \label{eqn:KKT2} \\
\mu(\sigma^2(1-q)-1+\rho))&=0. \label{eqn:KKT3}
\end{align}
The optimal solutions to satisfy the KKT conditions are
\begin{align}
q_*=\lambda, \quad \mu_*=\frac{\lambda}{1-\rho}, \quad \sigma^2_*=\frac{1-\rho}{1-\lambda}.
\end{align}
Since the KKT conditions are satisfied by $q_*,\sigma^2_*$ and $\mu_*$ then strong duality holds, thus 
\begin{align}
&\min\limits_{(\sigma^2,q) \in \mathcal{D}_{\rho}} f(\lambda,\sigma^2,q) \label{eqn:mlproof3star}\\
&=\max_{\mu} \min\limits_{\sigma^2,q} f(\lambda,\sigma^2,q) + \mu(\sigma^2(1-q)-1+\rho)) \\
&= f(\lambda,\frac{1-\rho}{1-\lambda},\lambda) \\
&= \frac{1}{2} \log{\frac{1}{1-\lambda^2}}-\frac{\lambda}{2} \log{(2\pi e)^2\frac{(1-\rho)^2(1+\lambda)}{1-\lambda}}. \label{eqn:mlprooffinal}
\end{align}
By combining (\ref{eqn:mlproof1}), (\ref{eqn:mlproof2}), (\ref{eqn:mlproof3star}) and (\ref{eqn:mlprooffinal}) we get the desired lower bound. 

\item For the case $\rho < q$ we omit the details due to lack of space. The optimal solutions are 
\begin{align}
q_*=\rho, \quad \sigma^2_*=\frac{1+\rho}{1+q}.
\end{align}

To conclude we show that $f(\lambda,\frac{1-\rho}{1-\lambda},\lambda) \leq f(\lambda,1,\rho)$ for $\lambda \leq \rho$. The argument goes through also for the case when $\rho$ is negative, which completes the proof.
\end{itemize}

\section*{Acknowledgment}
This work was supported in part by the Swiss National Science Foundation under Grant 169294. 


\bibliographystyle{IEEEtran}
\bibliography{IEEEabrv,wyner_lower}

\end{document}